\documentclass[%
reprint,prl,twocolumn,superscriptaddress,amsmath,amssymb,
 aps]{revtex4-2}

\usepackage{graphicx}
\usepackage{xcolor}
\usepackage{amsthm,bbm}
\usepackage{bm,dsfont}% bold math
\usepackage{braket}% bold math
\usepackage{array}
\usepackage{amssymb}
\usepackage{comment}
\usepackage[colorlinks=true, citecolor=blue, urlcolor=blue]{hyperref}
\usepackage{enumitem}

\makeatletter
\newtheorem*{rep@theorem}{\rep@title}
\newcommand{\newreptheorem}[2]{%
\newenvironment{rep#1}[1]{%
 \def\rep@title{#2 \ref{##1}}%
 \begin{rep@theorem}}%
 {\end{rep@theorem}}}
\makeatother
\usepackage{mathrsfs}

\newreptheorem{theorem}{Theorem}

\newtheorem{definition}{Definition}

\newtheorem*{result}{Result}

\newtheorem{observation}{Observation}
\newreptheorem{lemma}{Lemma}

\newcommand{\be}{\begin{equation}} 
\newcommand{\ee}{\end{equation}}
\newcommand{\beq}{\begin{eqnarray}}
\newcommand{\eeq}{\end{eqnarray}}
\newcommand{\rank}{\mathrm{rank}}
\newcommand{\id}{\mathds{1}}
\newcommand{\tr}{{\rm tr}}

\newcommand{\B}{\mathcal{B}}
\newcolumntype{C}{>{$}c<{$}}

\begin{document}

%%%%%%%%%%%%%%%%%%%%%%%%%%%%%%%%%%%%%%%%%%%%%%%%%%%%%%%%%%%%%%%%%%%

\title{Self-testing ideal quantum measurements}

%%%%%%%%%%%%%%%%%%%%%%%%%%%%%%%%%%%%%%%%%%%%%%%%%%%%%%%%%%%%%%%%%%%

\author{Debashis Saha}
\email{dsaha@iitbbs.ac.in}
\affiliation{Department of Physics, School of Basic Sciences, Indian Institute of Technology Bhubaneswar, Odisha 752051, India}
\author{Zhen-Peng Xu}
\email{zhen-peng.xu@ahu.edu.cn}
\affiliation{School of Physics, Anhui University, 230601 Hefei, People’s Republic of China}

\author{Ad\'an~Cabello}
\email{adan@us.es}
\affiliation{Departamento de F\'{\i}sica Aplicada II, Universidad de Sevilla, E-41012 Sevilla,
Spain}
\affiliation{Instituto Carlos~I de F\'{\i}sica Te\'orica y Computacional, Universidad de
Sevilla, E-41012 Sevilla, Spain}

%%%%%%%%%%%%%%%%%%%%%%%%%%%%%%%%%%%%%%%%%%%%%%%%%%%%%%%%%%%%%%%%%%%

\begin{abstract}
A fundamental prediction of quantum mechanics is that observables can be measured through ideal measurements, that is, measurements that give the same outcome when repeated and do not disturb any observable that is jointly measurable. In terms of quantum instruments, such an ideal measurement is described by projective effects together with the associated L\"uders state-update rule. Here, we show that quantum mechanics allows for self‑testing \textit{any} ideal measurement. Using solely the output statistics of bipartite Bell tests in which a non-demolition measurement is sometimes added, it is possible to certify in a device-independent manner whether the added measurement is ideal. %The proof relies on quantum perfect strategies.
\end{abstract}

%%%%%%%%%%%%%%%%%%%%%%%%%%%%%%%%%%%%%%%%%%%%%%%%%%%%%%%%%%%%%%%%%%%

\maketitle

%%%%%%%%%%%%%%%%%%%%%%%%%%%%%%%%%%%%%%%%%%%%%%%%%%%%%%%%%%%%%%%%%%%

{\em Introduction.---} One of the most powerful concepts in quantum information theory is self-testing \cite{Yao_self,vsupic2019self}: the ability to certify whether specific devices perform particular quantum operations (such as state preparations and measurements of observables), without requiring any knowledge of the internal workings of the devices. Instead, certification relies solely on the outcome statistics of experiments conducted using these devices.

One of the most fundamental predictions of quantum mechanics is the existence of minimally disturbing measurements, also called ``ideal'' or ``sharp'' \cite{Pokorny:PRL20}. That is, 
measurements that yield identical outcomes upon repetition and leave unchanged the outcome statistics of any observable that is jointly measurable (compatible) with them \cite{Chiribella:2016IC,Kleinmann_2014}. Ideal measurements have some fundamental properties: (i) If the input state is pure then, once you read the outcome, the output state is pure \cite{Luders:1951APL}. 
(ii) Even if the input state is maximally mixed, if the observable is maximal, the output state is pure.
(iii)~Enabling outcome-dependent delayed refinements: one can measure any non-maximal observable and {\em later} decide what compatible observable one wants to measure next, depending on, e.g., the outcome obtained (rather than having to decide beforehand to measure a maximal observable).
(iv)~Enabling contextual correlations. Ideal measurements allow one to separate the correlations that are inherently permitted by the theory from the disturbance that may be introduced by a clumsy measurement of the observable.
This way, ideal measurements enable sequential experiments which can unveil fundamental correlations which cannot occur in Bell tests and which are crucial to understand quantum theory \cite{Cabello:2013PRL,Cabello:2019PRA}.
From a broader perspective, it has been proven that any general Bayesian theory for experiments on individual physical systems {\em requires} ideal measurements \cite{Chiribella:2020PRR}.

As pointed out by Kip Thorne in \cite{Braginskii:1992Book}, until recently, ``the only feature of the quantum theory of measurement which entered into the experiments was the probability interpretation of squared amplitudes. However, 
technology finally began to catch up with the (\ldots) concepts [of the quantum theory of measurement, allowing for] repeated measurements on single quantum systems'' \cite{Braginskii:1992Book}. So while even modern books say that
``repeatable tests (\ldots) exist mostly in the imagination of theorists'' \cite{Peres:1993},
we can nowadays perform 53 million consecutive tests on the same trapped-ion qutrit with a repeatability above 99.6\% \cite{Leupold:2018PRL}.

%In Chapter 4 of his {\em Mathematical Foundations of Quantum Mechanics} \cite{vonNeumann:1932SPR}, von Neumann wrote that when an observable with degenerate spectrum is measured on a pure state and the outcome corresponding to the degenerate eigenvalue is obtained, the post-measurement state should be mixed. However, L\"uders observed that, to be {\em mathematically consistent}, quantum mechanics requires a different rule for quantum state transformation \cite{Luders:1951APL}: one that preserves purity. 

%Interestingly, from a theory-independent operational perspective, the rule proposed by L\"uders \cite{Luders:1951APL} for state-update rule implies the existence of measurements that are minimally disturbing (or ideal or sharp). This discovery is crucial for understanding what class of theory quantum mechanics is, as 
%That is, regardless of whether their actual implementation is technologically challenging and even if they requires infinite resources \cite{Guryanova:2020Q}, minimally disturbing measurements (and pure states) are central to science. 

Here, we address the question of whether quantum ideal measurements can be self-tested. That is, whether quantum mechanics allows one to certify, given only a black-box device, that the device implements an ideal measurement.

Self-testing quantum ideal measurements is important from a practical perspective: many quantum technologies involve long sequences of measurements where each step produces classical information while also preparing a post-measurement quantum state which is used in subsequent stages. Examples include extensions of measurement-based quantum computation \cite{Raussendorf:PRL2001} based on sequential ideal measurements~\cite{BSS16}, quantum processing tasks using sequential measurements~\cite{Leupold:2018PRL,kleinmann2011memory}, and quantum communication networks where intermediate parties receive quantum systems, perform measurements to obtain partial information, and forward the resulting quantum states to other nodes~\cite{Brierley_prl2015,Savov:12,Uzi:2025}. In these scenarios, the correct functioning of the device implementing an ideal measurement is crucial. 

In quantum theory, an instrument is the mathematical object that describes both the classical outcome produced by a measurement device and the corresponding post-measurement quantum state~\cite{Davies:JMA70,Ozawa:JMP1984,Dressel:PRA2013}. Experimental techniques for characterizing quantum instruments already exist, typically based on variants of quantum process tomography~\cite{Pokorny:PRL20,Stricker:PRX2022}. The methods we introduce here differ fundamentally; here, the certification relies only on device-independent statistics, namely the observed conditional probabilities of outputs given inputs. Consequently, our protocol can be applied even when the internal functioning of the measurement device is unknown or intentionally hidden, for example, when the device is provided by an untrusted manufacturer. 
While device-independent certification of quantum channels~\cite{Sekatski_prl_2018} and non-projective measurements~\cite{Wagner2020} has recently been explored, our main objective here is the self-testing of ideal measurements and their induced state transformations.

From a conceptual standpoint, there are two principal approaches to reconstructing quantum theory from operational principles. The first is the general probabilistic theories framework \cite{Janotta_2014}, which puts the emphasis on states (classes of equivalence of preparations with respect to measurements), observables (classes of equivalence of measurements with respect to preparations), and transformations. The second is the device-independent approach \cite{Scarani:2013APS}, which focuses on conditional probabilities relating measurement inputs to outputs, treating the devices as ``black boxes.'' Traditionally, these two programs have developed independently. However, it has been pointed out that bridging them offers significant conceptual and technical benefits \cite{Chiribella:2016IC}. The methods we propose here serve as a concrete example of such benefits: they operate within the device-independent paradigm by relying solely on input-output probability distributions, yet they go beyond standard DI results by certifying not only states and observables but also state transformations.

%%%%%%%%%%%%%%%%%%%%%%%%%%%%%%%%%%%%%%%%%%%%%%%%%%%%%%%%%%%%%%%%%%%

\textit{The basic idea.---}
The basic idea of our protocols is illustrated in Fig.~\ref{Fig1}. Two distant parties, Alice and Bob, share uncharacterized measurement devices with classical inputs $x$ and $y$ and outputs $a$ and $b$, respectively. In each round of the experiment, they randomly choose inputs and record the corresponding outputs. A third party, Carol, who is spacelike separated from Alice, may choose to insert an additional measurement device before Bob's device. When Carol's device is used, it produces a classical outcome $c$ and prepares an unknown post-measurement quantum state that is subsequently measured by Bob.
In the rounds where Carol does not insert her device, the experiment reduces to a standard Bell test, allowing Alice and Bob to obtain correlations $p(a,b|x,y)$. In the rounds where Carol performs a measurement and obtains outcome $c$, the correlations become $p(a,b,c|x,y)$. 

The correlations $\{p(a,b|x,y)\}$ are first used to self-test the shared state between Alice and Bob, and Alice's measurements (up to standard local degrees of freedom) through the maximal violation of appropriate Bell inequalities \cite{Yao_self}. Because Carol’s choice of whether to perform a measurement is made at a spacelike-separated location, it cannot influence Alice's measurement. Consequently, Alice's observables remain the same in both types of rounds.
Conditioning on Carol's outcome $c$, the correlations $\{p(a,b|x,y,c)\}$ can again exhibit maximal Bell inequality violations between Alice and Bob. Together with the observed statistics $\{p(a,c|x)\}$ and the previously deduced self-tested state and Alice's measurements, allow one to self-test the post-measurement state given $c$, up to the device-independent equivalences. 
In this way, the protocol self-tests an ideal quantum measurement implemented in Carol's instrument. 

%%%%%%%%%%%%%%%%%%%%%%%%%%%%%%%%%%%%%%%%%%%%%%%%%%%%%%%%%%%%%%%%%%%
% Fig. 1
%%%%%%%%%%%%%%%%%%%%%%%%%%%%%%%%%%%%%%%%%%%%%%%%%%%%%%%%%%%%%%%%%%%

\begin{figure}[!t]
    \centering  \includegraphics[scale=0.51]{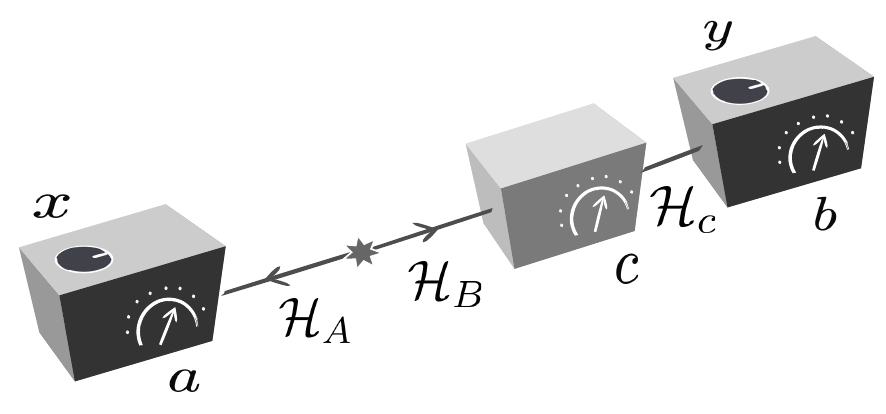}
    \caption{Extended Bell experiment involving Alice, Bob, and Carol. Alice and Bob perform measurements with inputs $x,y$ and outputs $a,b$, respectively. Carol randomly decides whether to place her black-box instrument, producing outcome $c$ and an unknown post-measurement system, before Bob's device. The experiment therefore produces correlations $\{p(a,b|x,y)\}$ and $\{p(a,b,c|x,y)\}$, from which Carol's instrument can be self-tested to be sharp in a fully device-independent manner.}
    \label{Fig1}
\end{figure}

%%%%%%%%%%%%%%%%%%%%%%%%%%%%%%%%%%%%%%%%%%%%%%%%%%%%%%%%%%%%%%%%%%%

We now proceed to mathematically formulate the self-testing of ideal measurements precisely. 

%%%%%%%%%%%%%%%%%%%%%%%%%%%%%%%%%%%%%%%%%%%%%%%%%%%%%%%%%%%%%%%%%%%

\textit{Self-testing ideal measurements.---}
In the Bell experiment described above, let the unknown shared state $\rho$ belong to $\mathcal{L}\left(\mathcal{H}_A\otimes \mathcal{H}_B\right)$. Alice's and Bob's devices implement unknown measurements $\{A_{a|x}\}_{a,x}$ and $\{B_{b|y}\}_{b,y}$ acting on $\mathcal{H}_A$ and $\mathcal{H}_B$, respectively. 

The unknown measurement instrument implemented by Carol is denoted by
$\{\mathcal{M}_c\}_c$. For every outcome $c$,
$\mathcal{M}_c:\mathcal{L}(\mathcal{H}_B)
    \longrightarrow \mathcal{L}(\mathcal{H}_c)$
is a completely positive trace-nonincreasing map, and
$\sum_c\mathcal{M}_c$ is trace preserving. 
Here,
$\mathcal{H}_c$ denotes the arbitrary post-measurement quantum system
delivered to Bob when Carol obtains outcome $c$. No particular
relation between the input Hilbert space $\mathcal{H}_B$ and the
output Hilbert spaces $\mathcal{H}_c$ is assumed. In particular,
Carol's instrument may involve arbitrary internal ancillary systems
or environments that are not accessible to Bob.

Consider now the reference ideal measurement that we aim to self-test, described by a collection of orthogonal projectors
$\{\Pi_c\}_c$ acting on $\mathbb{C}^d$, with
$\sum_c\Pi_c=\id$. The corresponding ideal measurement instrument is
the L\"uders instrument
$\mathcal{P}_c(X)=\Pi_c X\Pi_c .$
The Choi--Jamio{\l}kowski isomorphism provides a natural way of
formulating the certification of such an instrument. Let
$\ket{\phi_d^+} =\frac{1}{\sqrt d}\sum_{i=0}^{d-1}\ket{ii}.$
For a quantum operation $\mathcal{N}$ acting on a $d$-dimensional
input system, its Choi operator is $J(\mathcal{N}) :=
(\id\otimes\mathcal{N}) \bigl(\ket{\phi_d^+}\!\bra{\phi_d^+}\bigr).$
The Choi--Jamio{\l}kowski isomorphism implies that $\mathcal{N}$ is
completely characterized by $J(\mathcal{N})$. In particular, the
Choi operator corresponding to the ideal operation
$\mathcal{P}_c$ is
$J(\mathcal{P}_c) =(\id\otimes\Pi_c)\ket{\phi_d^+}\!\bra{\phi_d^+}(\id\otimes\Pi_c).$
Notice that this operator is in general unnormalized, with
$\tr[J(\mathcal{P}_c)] =\frac{\rank(\Pi_c)}{d}.$

In a device-independent setting, the physical input and output Hilbert spaces of the unknown instrument element $\mathcal{M}_c$ are not characterized and need not coincide with those of the reference instrument $\mathcal{P}_c$. Comparing the two operations therefore requires suitable transformations that identify the relevant physical input and output subsystems with their ideal counterparts.
At the level of the Choi representation, this can be viewed as a pre-processing operation $V$ that maps the ideal input subsystem,
together with suitable auxiliary degrees of freedom, to the physical input of $\mathcal{M}_c$, followed by a post-processing operation
$U_c$ that extracts the ideal post-measurement subsystem from the physical output of $\mathcal{M}_c$.
More precisely, the equivalence
between the physical and ideal instrument elements can be expressed as
\begin{equation}
\begin{split}
    &(\id_A\otimes\mathcal{U}_c)
    (\id_A\otimes\mathcal{M}_c)
        (\id_A\otimes\mathcal{V})
            \Big[
                \ket{\phi_d^+}\!\bra{\phi_d^+}
                \otimes\rho_{\mathrm{aux}}
            \Big]
    \\
    &\quad =
    (\id\otimes\Pi_c)
    \ket{\phi_d^+}\!\bra{\phi_d^+}
    (\id\otimes\Pi_c)
    \otimes\rho_{\mathrm{aux},c}.
    \label{CJst}
\end{split}
\end{equation}
where $\mathcal{V}(X)=VXV^\dagger,
\mathcal{U}_c(X)=U_cXU_c^\dagger,$
and $\rho_{\mathrm{aux}}$ and $\rho_{\mathrm{aux},c}$ denote the
input and output auxiliary states, respectively.
This pre- and post-processing
picture is illustrated in Fig.~\ref{Fig2}. 

%%%%%%%%%%%%%%%%%%%%%%%%%%%%%%%%%%%%%%%%%%%%%%%%%%%%%%%%%%%%%%%%%%%
% Fig. 2
%%%%%%%%%%%%%%%%%%%%%%%%%%%%%%%%%%%%%%%%%%%%%%%%%%%%%%%%%%%%%%%%%%%

\begin{figure}[!t]
    \centering
    \includegraphics[scale=0.49]{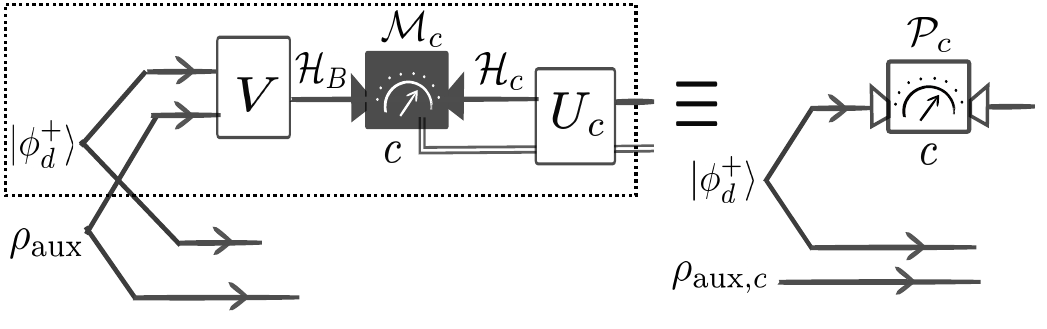}
    \caption{Schematic representation of instrument self-testing via the Choi--Jamio{\l}kowski isomorphism. The pre-processing transformation $V$ identifies the physical input of the unknown operation $\mathcal{M}_c$, while $U_c$ extracts the ideal subsystem from its post-measurement output. In the device-independent realization, $V$ is fixed by the initial state self-test as $V=U_B^\dagger$.}
    \label{Fig2}
\end{figure}

%%%%%%%%%%%%%%%%%%%%%%%%%%%%%%%%%%%%%%%%%%%%%%%%%%%%%%%%%%%%%%%%%%%

The relation in Eq.~\eqref{CJst}, however, is not yet device independent, since it assumes access to the maximally entangled input
$\ket{\phi_d^+}$.
We therefore first use the correlations $\{p(a,b|x,y)\}$, obtained
in the absence of Carol's instrument, to self-test the maximally entangled state. More precisely, there exist local extraction isometries $U_A$ and $U_B$ such that Eq.~\eqref{stAB} holds. Consequently, the pre-processing transformation appearing in
Eq.~\eqref{CJst} can be identified as $V=U_B^\dagger$, while the
post-processing $U_c$ extracts the ideal subsystem from Carol's
post-measurement output.

In particular, when Carol applies her instrument and obtains outcome $c$, the corresponding unnormalized post-measurement state shared by Alice and Bob is $\rho_{A B_c} := (\id_A\otimes\mathcal{M}_c) \bigl(\rho\bigr),$ and $p(c)=\tr[\rho_{A B_c}]$ is the probability of outcome $c$. For each
outcome $c$, an additional local extraction isometry $U_c$ is then used to identify the post-measurement system delivered by Carol to Bob with the output subsystem of the reference L\"uders operation. Importantly, $U_c$ acts only on the system accessible to Bob and has no access to any environment retained or discarded by Carol.

We can therefore formulate instrument self-testing directly in terms of the unknown
state $\rho$ through two relations. The first relation identifies the input state through the standard self-testing \eqref{stAB}. While the second relation, which follows from \eqref{CJst} and \eqref{stAB}, certifies the complete outcome operation by comparing its unnormalized post-measurement state with the Choi operator of the corresponding ideal L\"uders operation. Since the latter is unnormalized, this comparison simultaneously captures the probability associated with the measurement outcome and the
corresponding post-measurement update.

This motivates the following definition.

\begin{definition}[Self-testing of an ideal measurement instrument]
\label{def:instrument-self-testing}
Let $\{\mathcal{M}_c\}_c$ be the unknown measurement instrument
implemented by Carol, and let
$\{\mathcal{P}_c\}_c$, with
$\mathcal{P}_c(X)=\Pi_cX\Pi_c$, denote the reference ideal
L\"uders instrument. We say that the ideal measurement instrument
$\{\mathcal{P}_c\}_c$ is self-tested if the following conditions
hold.

\begin{itemize}[noitemsep]

\item There exist local extraction isometries $U_A$ and $U_B$ such
that
\begin{equation}
    (\mathcal{U}_A\otimes \mathcal{U}_B)(\rho)
    = \ket{\phi_d^+}\!\bra{\phi_d^+}\otimes \rho_{\mathrm{aux}}.
    \label{stAB}
\end{equation}

\item For every outcome $c$, there exists a local extraction
isometry $U_c$, acting only on the post-measurement system
$\mathcal{H}_c$ accessible to Bob, such that, with the same extraction isometry $U_A$ as in
Eq.~\eqref{stAB},
\begin{equation}
\begin{split}
    &(\mathcal{U}_A\otimes\mathcal{U}_c)
    \left[
        (\id_A\otimes\mathcal{M}_c)
        \bigl(\rho\bigr)
    \right]
    \\
    = \ &
    (\id\otimes\Pi_c)
    \ket{\phi_d^+}\!\bra{\phi_d^+}
    (\id\otimes\Pi_c)
    \otimes\rho_{\mathrm{aux},c},
    \label{stABcall}
\end{split}
\end{equation}
where
$\mathcal{U}(X)=UXU^\dagger$, and
$\rho_{\mathrm{aux}}, \rho_{\mathrm{aux},c}$ are arbitrary auxiliary states.
\end{itemize}
The isometry $U_c$ acts exclusively on the quantum system delivered
by Carol to Bob and has
no access to any internal environment that may be retained or discarded by Carol.
\end{definition}

It is noteworthy that the above definition of instrument self-testing allows Carol's device to communicate the outcome $c$ to Bob's device, so that Bob's measurement $B^{(c)}_{b|y}$ may depend explicitly on $c$. 
%This is in contrast to the framework considered in \cite{Wagner2020}, where Bob's measurement operators are required to be independent of Carol's outcome.
Now we point out one important observation.

\begin{observation}[Self-testing of a rank-one instrument element]
\label{lem:rank-one-instrument}
Let $\mathcal{M}_c$ be an outcome operation of Carol's
instrument and let the target effect be the rank-one projector $\Pi_c=\ket{i_c}\!\bra{i_c}.$ Suppose that the initial correlations self-test the maximally entangled state \eqref{stAB} and Alice's computational-basis measurement $x_0$, and suppose that the Alice--Carol correlations satisfy
\begin{equation}
    p(a,c|x_0)
    =
    \frac{1}{d}\delta_{a,i_c} .
    \label{eq:rank-one-AC}
\end{equation}
Then there exists a local extraction isometry $U_c$, such that Eq.~\eqref{stABcall} holds with $\Pi_c=\ket{i_c}\!\bra{i_c}.$ Thus, for a rank-one outcome, no additional self-test of the
post-measurement state is required.
\end{observation}
While the detailed proof is deferred to the Appendix, the key point is that, for a rank-one outcome, the effect $\Pi_c$ completely determines the input dependence of the corresponding operation. The only remaining freedom lies in the outcome-dependent output state, which can be absorbed into the auxiliary degrees of freedom by an appropriate local extraction isometry. For $\rank(\Pi_c)\geq2$, this is no longer the case. Two instruments
may have the same effect $\Pi_c$ while acting differently within $\operatorname{supp}(\Pi_c)$. 
For example, consider the rank-two projector
$\Pi_c=\ket{0}\!\bra{0}+\ket{1}\!\bra{1}.$
The ideal L\"uders operation
$\mathcal{P}_c(\rho)=\Pi_c\rho\Pi_c $ and the completely dephasing operation
$\mathcal{M}_c(\rho)= \ket{0}\!\bra{0}\rho\ket{0}\!\bra{0} + \ket{1}\!\bra{1}\rho\ket{1}\!\bra{1}$
have the same effect,
and therefore give identical probabilities for outcome $c$. Nevertheless, their post-measurement actions are inequivalent. When acting on one subsystem of a maximally entangled state,
the former produces the maximally entangled
state whereas the latter produces the separable state. This illustrates why, for projectors of rank at least two, self-testing of the post-measurement state is essential. 

%%%%%%%%%%%%%%%%%%%%%%%%%%%%%%%%%%%%%%%%%%%%%%%%%%%%%%%%%%%%%%%%%%%

\textit{Main result.---}
A Kochen–Specker KS set is a finite set of rank-one projectors or vectors ${\cal K}$ in $\mathbbm{C}^d$ of finite dimension $d \ge 3$, which does not admit an assignment $f: \mathcal{K} \rightarrow \{0,1\}$ satisfying the following conditions. (I)~Two mutually orthogonal vectors $v_i$ and $v_j$ cannot simultaneously satisfy $f(v_i)=f(v_j)=1$. (II)~For every set of mutually orthogonal vectors $\{v_i\}$ that forms a basis, one must have $\sum_i f(v_i)=1$. A KS set is said to be \textit{complete} if every pair of orthogonal vectors belongs to at least one basis. One of the well-known examples of a complete KS set is the Peres-24 set of vectors in dimension $4$.

For any complete KS set, there exists a Bell game, called the context-projector game (defined through the orthogonality graph of the respective KS set, see Appendix), whose perfect winning condition can be realized by quantum theory, while classical (or local realist) theories cannot achieve it. Ref.~\cite{Xu_PRL_2024} shows that by merging $(d-3)$ sets of Peres-24 vectors one can construct a complete KS set in any dimension $d\geq 4$. This merging is performed by placing Peres-24 vectors (see Tab. \ref{tab:peres24} in Appendix) in sequential $4$-dimensional subspaces within dimension $d$ and removing the common vectors. 

More importantly, Ref.~\cite{Xu_PRL_2024} shows that for such KS sets in $\mathbbm{C}^d$, the respective perfect strategies for theBell game (denoted by $\mathcal{B}^\text{P}_d$) self-test the KS projectors on Bob's side, the observables corresponding to the bases formed by the KS vectors on Alice's side, and the maximally entangled state of dimension $d$. With this background in place, we now present our final main result.
\begin{result}[Self-testing ideal measurements from perfect strategies]\label{result-3}
    Ideal L\"uders measurement instruments with an arbitrary number of outcomes and projectors of arbitrary rank can be self-tested from perfect strategies of context-projector games.
\end{result}
While the detailed proof is deferred to the Appendix, its underlying idea is the following. We first self-test the maximally entangled state $\ket{\phi_d^+}$ and the measurements associated with the KS set of vectors. The ideal post-measurement state corresponding to any given outcome by Carol's measurement is a maximally entangled state supported on the respective subspace of the projector. Owing to the construction of the \textit{merged Peres-24} set, the KS vectors contained entirely within that subspace themselves form a complete KS set, provided that the subspace is at least four-dimensional. Therefore, the same measurements restricted to that subspace reproduce the perfect quantum correlations of the associated context-projector game. By applying the self-testing result to these correlations, one certifies the post-measurement state and consequently the projector for every outcome on Carol's measurement. For projectors spanning two- or three-dimensional subspaces, we introduce additional measurements on Bob's side and use Bell--Clauser--Horne--Shimony--Holt (Bell--CHSH) self-testing \cite{Yao_self,mckague2012robust,bamps2015sum,Coladangelo_NatCom} within suitable subspaces to certify the post-measurement state and, consequently, the corresponding projector.

To illustrate a simple example, the complete KS set in dimension $8$ constructed from the Peres-24 set is given in Appendix. The corresponding context-projector game can be used to self-test the two-outcome ideal measurement $\{\sum_{i=0}^3\ket{i}\!\bra{i}, \sum_{i=4}^7\ket{i}\!\bra{i}\}$.
 
It is important to note that the self-testing scheme can be readily adapted to self-test an individual projector with L\"uders rule of arbitrary dimension $d$. In this case Alice and Bob first self-test $\ket{\phi^+_{d+1}}$ and consider a binary-outcome measurement on Carol's system described by $\{ M_0,\id-M_0\}$, where the target projector $\Pi = \sum_{i=0}^{d-1}\ket{i}\!\bra{i} $ is associated with outcome $M_0$, and $\id-M_0$ is self-tested to the rank-one projector.

Finally, in the Appendix, we discuss how the robustness of instrument self-testing can be quantified in terms of the robustness of the two self-testing relations appearing in the definition~\ref{def:instrument-self-testing}.

%%%%%%%%%%%%%%%%%%%%%%%%%%%%%%%%%%%%%%%%%%%%%%%%%%%%%%%%%%%%%

\textit{Conclusions.---} Despite their necessity for the consistency of the theory, as noted by L\"uders \cite{Luders:1951APL}, the presence and fundamental role of ideal measurements was largely overlooked by the founding fathers of quantum mechanics \cite{vonNeumann:1932SPR}. In fact, recognizing their significance is a relatively recent achievement, stemming from L\"uders  correction \cite{Luders:1951APL} of an error in von Neumann's {\em Mathematical Foundations of Quantum Mechanics} \cite{vonNeumann:1932SPR} (Chapter IV), and later from the operational interpretation of L\"uders' state update rule by Chiribella and Yuan \cite{chiribella2014measurement} and by Kleinmann \cite{kleinmann2014sequences}: in quantum mechanics, Hermitian operators represent {\em the minimally disturbing measurement of} observables (whereas positive operator-valued measures represent all possible ways of measuring observables, regardless of the state transformation they produce). While recent efforts have focused on proving that these mythical objects actually exist through tomographical reconstructions \cite{Pokorny:PRL20}, we have taken a further step by proving that quantum mechanics allows us to certify ideal measurements based solely on the correlations produced by a bipartite system in which a third party sometimes introduces an ideal measurement. We believe this is a fundamental result. On a conceptual level, it shows the benefits of unifying the device‑independent and operational approaches to the foundations of quantum theory. On a practical level, it has a wide range of useful applications and may stimulate a new experiments.

%%%%%%%%%%%%%%%%%%%%%%%%%%%%%%%%%%%%%%%%%%%%%%%%%%%%%%%%%%%%%%%%%%%	

\paragraph{Acknowledgments.} Z.P.X. is supported by {National Natural Science Foundation of China} (Grant No.~12305007, No.~12675010),
Anhui Provincial Natural Science Foundation (Grant No.~2508085Y003). A.C.\ was supported by the AEI/MICIU Project ``Quantum Correlations for Quantum Information and Computation'' (QCQIC) (Project No.~PID2025-174467NB-I00). The use of generative artificial intelligence tools was limited to language editing. 

%%%%%%%%%%%%%%%%%%%%%%%%%%%%%%%%%%%%%%%%%%%%%%%%%%%%%%%%%%%%%%%%%%%	

%%%%%%%%%%%%%%%%%%%%%%%%%%%%%%%%%%%%%%%%%%%%%%%%%%%%%%%%%%%%%%%%%%%

\bibliography{bibliography}

%%%%%%%%%%%%%%%%%%%%%%%%%%%%%%%%%%%%%%%%%%%%%%%%%%%%%%%%%%%%%%%%%%%

%\newpage
%\onecolumngrid
%\appendix 
\section{Appendix} 

%%%%%%%%%%%%%%%%%%%%%%%%%%%%%%%%%%%%%%%%%%%%%%%%%%%%%%%%%%%%%%%%%%

\subsection{Proof of Observation \ref{lem:rank-one-instrument}}

%%%%%%%%%%%%%%%%%%%%%%%%%%%%%%%%%%%%%%%%%%%%%%%%%%%%%%%%%%%%%%%%%%

\begin{proof}[Proof of Observation \ref{lem:rank-one-instrument}]
Let $\rho_{A B_c}:=(\id_A\otimes\mathcal{M}_c)
\bigl(\rho\bigr)$ denote the unnormalized post-measurement state corresponding to Carol's outcome $c$. Its trace is the probability of obtaining that outcome. From Eq.~\eqref{eq:rank-one-AC}, we have $p(c) = \sum_a p(a,c|x_0) = \frac{1}{d}.$
We assume that the initial correlations self-test the maximally
entangled state and Alice's computational-basis measurement.
Hence, under the extraction isometry $U_A$, Alice's measurement
operators for $x_0$ act as
$\{\ket{a}\!\bra{a}\}_{a=0}^{d-1}$ on the extracted ideal
$d$-dimensional subsystem. Define
\begin{equation}
    \widehat{\rho}_c
    := (\mathcal{U}_A\otimes\id)
    \bigl(\rho_{A B_c}\bigr),
\end{equation}
where $\mathcal{U}_A(X)=U_AXU_A^\dagger$. Denoting Alice's
extracted ideal subsystem by $A_0$, the observed Alice--Carol
correlations \eqref{eq:rank-one-AC} imply
\begin{equation}
    \tr\!\left[  \left(\ket{a}\!\bra{a}_{A_0}\otimes\id\right)
        \widehat{\rho}_c
    \right]
    =
    \frac{1}{d}\delta_{a,i_c}.
    \label{eq:rank-one-proof-prob}
\end{equation}
Let $\widehat{\rho}_c^{A_0}
:= \tr_{\mathrm{rest}}[\widehat{\rho}_c]$
be the reduced state on Alice's extracted ideal subsystem. Since
$\widehat{\rho}_c^{A_0}\geq0$, Eq.~\eqref{eq:rank-one-proof-prob}
implies $\bra{a}\widehat{\rho}_c^{A_0}\ket{a}=0$ for all $a\neq i_c.$ Positivity then implies that all matrix elements involving such
$\ket{a}$ vanish. Therefore,
\begin{equation}
    \widehat{\rho}_c^{A_0}
    =
    \frac{1}{d}
    \ket{i_c}\!\bra{i_c}.
    \label{eq:rank-one-A-marginal}
\end{equation}

Since the reduced state in Eq.~\eqref{eq:rank-one-A-marginal} has
rank one, the full post-measurement state necessarily factorizes as
\begin{equation}
    \widehat{\rho}_c
    =
    \frac{1}{d}
    \ket{i_c}\!\bra{i_c}_{A_0}
    \otimes \sigma_c,
    \label{eq:rank-one-factorization}
\end{equation}
where $\sigma_c$ is a normalized state on the remaining degrees of
freedom, including Alice's auxiliary system and Carol's
post-measurement output accessible to Bob.

It remains to construct Bob's extraction isometry. Since the ideal
post-measurement state on Bob's extracted subsystem is the fixed
state $\ket{i_c}$, there exists a local isometry $U_c$, acting only
on the post-measurement system $\mathcal H_c$ together with a local
ancillary system, that introduces an ideal $d$-dimensional subsystem
$B_0$ in the state $\ket{i_c}$ and transfers all pre-existing output
degrees of freedom to the auxiliary subsystem. Consequently,
\begin{equation}\label{eq:rank-one-U}
(\mathcal{U}_A\otimes\mathcal{U}_c)
    \left[
        \rho_c^{A B_c}
        \otimes
        \ket{\mathrm{an}_c}\!\bra{\mathrm{an}_c}
    \right]
    =
    \frac{1}{d}
    \ket{i_c i_c}\!\bra{i_c i_c}
    \otimes\rho_{\mathrm{aux},c}
\end{equation}
for some normalized auxiliary state $\rho_{\mathrm{aux},c}$.

Finally, by noting
\begin{equation}
(\id\otimes\Pi_c)
    \ket{\phi_d^+}\!\bra{\phi_d^+}
    (\id\otimes\Pi_c)
    =
    \frac{1}{d}
    \ket{i_c i_c}\!\bra{i_c i_c}
\end{equation}
for $\Pi_c=\ket{i_c}\!\bra{i_c}$ and \eqref{eq:rank-one-U}, we obtain Eq.~\eqref{stABcall}. Hence the outcome
operation $\mathcal{M}_c$ is self-tested, in the sense of
Definition~\ref{def:instrument-self-testing}, to be the rank-one L\"uders operation associated with $\Pi_c$.
\end{proof}

%%%%%%%%%%%%%%%%%%%%%%%%%%%%%%%%%%%%%%%%%%%%%%%%%%%%%%%%%%%%%%%%%%

\subsection{Context-projector game and perfect quantum strategy}

%%%%%%%%%%%%%%%%%%%%%%%%%%%%%%%%%%%%%%%%%%%%%%%%%%%%%%%%%%%%%%%%%%

For any complete KS set, there exists a nonlocal game, known as the context-projector KS game, that admits a perfect quantum winning strategy but no perfect classical strategy \cite{aolita2012pra,Xu_PRL_2024}.
Let $G$ be the orthogonality graph of a complete KS set consisting of $n$ vectors $\{\ket{v_i}\}_{i=0}^{n-1}$ in $\mathbbm{C}^d$. Each vertex of $G$ represents a vector of the KS set, and two vertices are connected by an edge whenever the corresponding vectors are orthogonal. A clique in $G$ therefore represents a set of mutually orthogonal vectors. In particular, a maximal clique (or $d$-clique) is a clique that cannot be extended by adding another adjacent vertex; such cliques contain exactly $d$ vectors and correspond to orthonormal bases of $\mathbbm{C}^d$.
Let $T_x$ denote the $d$-cliques of $G$. Each clique contains $d$ distinct vertices,
$T_x := \{T_{x,0},\ldots, T_{x,d-1}\},$
where $x$ labels the $d$-clique and $a=0,\ldots,d-1$ labels the vertices within that clique. Suppose there are $m$ such $d$-cliques, indexed by $x=0,\ldots,m-1$.

The context-projector KS game is defined as follows. Alice receives a clique label, that is, $x\in \{0,\ldots,m-1\}$, and must output a vertex $T_{a,x}$ belonging to that clique, that is, $a\in \{0,\ldots,d-1\}$. Bob receives a vertex label $y\in \{0,\ldots,n-1\}$ and outputs a binary value $b\in \{0,1\}$. The players win the game if Bob outputs $1$ when Alice selects the vertex given to Bob, and outputs $0$ otherwise. 

Formally, the Bell expression in this game is given by
\be 
\B^{\text{P}}_d = \frac{1}{md} \sum_{a,b,x,y} c_{a,b,x,y} \ p(a,b|x,y), 
\ee 
where
\be 
c_{a,b,x,y} = 
\begin{cases}
1, \quad \text{ if } y \in T_x, \ y \neq T_{x,a} \text{ and } b=0, \\
1, \quad \text{ if } y \in T_x, \ y = T_{x,a} \text{ and } b=1, \\
0, \quad\; \text{otherwise.}
\end{cases}
\ee

By merging $(d-3)$ sets of Peres-24 vectors one can construct a complete KS set in any dimension $d\geq 4$ \cite{Xu_PRL_2024}. This merging is performed by placing Peres-24 vectors (see Tab. \ref{tab:peres24}) in sequential $4$-dimensional subspaces within dimension $d$ and removing the common vectors. 

For example, Peres-39 set of vectors in $\mathbbm{C}^5$ (given in the Supplementary Material of \cite{Xu_PRL_2024}) is obtained by merging two Peres-24 sets in the two four-dimensional subspaces in sequence as follows:
\[
\left\{\begin{pmatrix}
\bra{v_i} & 0
\end{pmatrix}^T\right\}_i \cup 
\left\{\begin{pmatrix}
0 & \bra{v_i}
\end{pmatrix}^T\right\}_i .
\]

For the constructions obtained from the \textit{merged Peres-24} sets in an arbitrary dimension $d$, it has been shown that achieving the perfect value $\mathcal{B}_d^\text{P}=1$ self-tests the underlying quantum realization \cite{Xu_PRL_2024}. In particular, the shared state is certified to be the maximally entangled state $\ket{\phi^+_d}$, Alice's measurements are self-tested to be the orthonormal bases defined by the KS set, $A_x = \{ \ket{v_{T_{x,0}}}, \dots, \ket{v_{T_{x,d-1}}}\}$, and Bob's measurements are self-tested to be $B_y = \{\id-|v_y\rangle\!\langle v_y|, |v_y\rangle\!\langle v_y|\}$.

\begin{widetext}
\begin{center}
\begin{table}[h!]
\begin{tabular}{CCCCCCCCCCCCCCCCCCCCCCCC}
\hline\hline 
v_0 & v_1 & v_2 & v_3 & v_4 & v_5 & v_6 & v_7 & v_8 & v_9 & v_{10} & v_{11} & v_{12} & v_{13} & v_{14} & v_{15} & v_{16} & v_{17} & v_{18} & v_{19} & v_{20} & v_{21} & v_{22} & v_{23} \\
\hline \rule{0pt}{3.2mm}
1 & 0 & 0 & 0 & 1 & 1 & 1 & 1 & 1 & 1 & 1 & \bar{1} & 1 & 1 & 0 & 0 & 1 & 1 & 0 & 0 & 1 & 1 & 0 & 0 \\
0 & 1 & 0 & 0 & 1 & 1 & \bar{1} & \bar{1} & 1 & 1 & \bar{1} & 1 & 1 & \bar{1} & 0 & 0 & 0 & 0 & 1 & 1 & 0 & 0 & 1 & 1 \\
0 & 0 & 1 & 0 & 1 & \bar{1} & 1 & \bar{1} & 1 & \bar{1} & 1 & 1 & 0 & 0 & 1 & 1 & 1 & \bar{1} & 0 & 0 & 0 & 0 & 1 & \bar{1} \\
0 & 0 & 0 & 1 & 1 & \bar{1} & \bar{1} & 1 & \bar{1} & 1 & 1 & 1 & 0 & 0 & 1 & \bar{1} & 0 & 0 & 1 & \bar{1} & 1 & \bar{1} & 0 & 0 \\
 \hline \hline
\end{tabular}
\caption{\label{tab:peres24} Peres-24. Each column represents the unnormalized vector $|v_i\rangle$, with $\bar{1}=-1$.
}
\end{table}
\end{center}
\end{widetext}

%%%%%%%%%%%%%%%%%%%%%%%%%%%%%%%%%%%%%%%%%%%%%%%%%%%%%%%%%%%%%%%%%%

\subsection{Proof of the Result}

%%%%%%%%%%%%%%%%%%%%%%%%%%%%%%%%%%%%%%%%%%%%%%%%%%%%%%%%%%%%%%%%%%

\begin{proof}[Proof of the Result\ref{result-3}]
Suppose we wish to self-test an ideal measurement $\{\Pi_c\}_c$ acting on a $d$-dimensional Hilbert space. By Observation~\ref{lem:rank-one-instrument}, rank-one outcomes are already covered. It therefore remains to consider the higher-rank outcomes. Suppose there are $r$ such outcomes. By appropriately relabeling them, we denote the corresponding projectors by $
\{\Pi_c\}_{c=0}^{r-1},$ where each $\Pi_c$ has rank $d_c\geq 2$. We consider the projectors in the computational basis,
\begin{equation} \label{Pic}
    \Pi_c=\sum_{i=g(c)}^{g(c+1)-1} \ket{i}\!\bra{i},
\end{equation}
where $g(c)=\sum_{k=0}^{c-1} d_k,
\ c=1,\ldots,r-1, $ and $g(0)=0$. The function $g(c)$ specifies the cumulative dimension of the projectors up to outcome $c-1$. In what follows, we focus on self-testing these higher-rank outcomes. 

Let $\mathcal{V}:=\{\ket{v_i}\}_i$ denote the $d$-dimensional \emph{merged Peres-24} set, chosen such that its first basis coincides with the computational basis, $\ket{v_i}=\ket{i}, i=0,\ldots,d-1.$ Suppose that Alice and Bob observe a perfect winning strategy in the corresponding \emph{context-projector} game, i.e.,
$\mathcal{B}^{\mathrm P}_d=1,$
from the observed correlations $\{p(a,b|x,y)\}$. By the self-testing result of Ref.~\cite{Xu_PRL_2024}, this immediately implies the self-testing of the $d$-dimensional maximally entangled state as expressed in Eq.~\eqref{stAB}. Specifically, there exist local extraction isometries $U_A$ acting on $\mathcal{H}_A$ and $U_B$ acting on $\mathcal{H}_B$ such that Eq.~\eqref{stAB} holds.

To proceed, we divide the next part of the proof into two cases:
(i) $d_c\geq 4$, and (ii) $d_c\in\{2,3\}$. \\

\textit{Case-(i): $d_c\geq 4$.} We now consider the subset of vectors in $\mathcal{V}$ whose support is entirely contained in the subspace
$\Pi_c$ given in \eqref{Pic}.
Denote this subset by
$\mathcal{V}_c$, namely,
$\ket{v_i}\in \mathcal{V}_c$ if and only if $\Pi_c\ket{v_i} = \ket{v_i}.$ By the construction of \textit{merged Peres-24} set, $\mathcal{V}_c$ is itself a merged Peres-24 set in dimension $d_c$. Consequently, one can define a context-projector game associated with the orthogonality graph generated by this vector set. The resulting orthogonality graph is a subgraph of the orthogonality graph corresponding to the original $d$-dimensional KS set. Therefore, the inputs and outputs of this reduced context-projector game form subsets of those appearing in $\mathcal{B}^{\mathrm P}_d$. Also, note that the subspace $\Pi_c$ has dimension at least four, since merged Peres-24 sets exist from four dimensions.

Next, consider the situation in which Carol performs her
measurement and obtains outcome $c$. The corresponding
normalized post-measurement state shared by Alice and Bob is
\begin{equation}
    \widehat{\rho}_{AB_c}
    := \frac{1}{p(c)}
    (\id_A\otimes\mathcal{M}_c)
    \bigl(\rho\bigr),
    \label{post-state-c}
\end{equation}
where $\mathcal{M}_c:\mathcal{L}(\mathcal{H}_B)
\rightarrow\mathcal{L}(\mathcal{H}_c)$ is the completely positive,
trace-nonincreasing operation associated with outcome $c$, and the observed outcome
probability is $p(c)=d_c/d$.

This construction enables Alice and Bob to verify perfect strategies
$\mathcal{B}^{\mathrm P}_{d_c}$
from the conditional correlations
$\{p(a,b|x,y,c)\}$ for every outcome $c$, by restricting attention to the corresponding subsets of inputs and outputs. A key feature of the argument is that Alice's measurements are precisely the same measurements that appear in the original correlations $\{p(a,b|x,y)\}$. Since Carol's measurement choice is space-like separated from Alice, it cannot influence Alice's measurement operators, which have already been self-tested via the extraction isometry $U_A$.

We emphasize that the conditional self-testing is performed separately for each outcome $c$. Accordingly, Bob's measurements used to generate the conditional correlations
$\{p(a,b|x,y,c)\}$ may depend on $c$, which may be represented by a family
$\{B^{(c)}_{b|y}\}_{b,y}$ acting on $\mathcal{H}_c$.
This is consistent with Definition~\ref{def:instrument-self-testing}, which allows an outcome-dependent extraction isometry $U_c$. In contrast, Alice's measurements are common to all conditional tests and are described by the same extraction $U_A$ fixed by the initial self-test.

Importantly, $\mathcal{B}^{P}_{d_c}$ is evaluated directly from the
conditional probabilities $\{p(a,b|x,y,c)\}$, without any further renormalization. Therefore, attaining the perfect value $\mathcal{B}^{P}_{d_c}=1$ also implies that the probabilities
of outcomes outside the reduced context-projector game vanish.
Thus, conditioned on Carol obtaining outcome $c$, the relevant correlations have support only on the outcomes appearing in the
reduced game. 
Consequently, the self-testing result of Ref.~\cite{Xu_PRL_2024}
implies that, for each outcome $c$, there exists a local extraction isometry $U_c$ acting on Carol's post-measurement system
$\mathcal{H}_c$,
such that
\begin{equation}
(\mathcal{U}_A\otimes\mathcal{U}_c)(\widehat{\rho}_{AB_c})
      = \ket{\phi^+_{d_c}}\!\bra{\phi^+_{d_c}}
   \otimes\rho_{\mathrm{aux},c},
\label{psib-st}
\end{equation}
where
$\rho_{\mathrm{aux},c}$ is an arbitrary auxiliary state. \\

\textit{Case (ii): $d_c\in\{2,3\}$.} Our goal is to show the same relation Eq.~\eqref{psib-st} when $d_c=2$ or $3$. In this case, we supplement Bob with four additional measurements. If there is no projector with $d_c= 3$ then only two additional measurements are required. We label these measurements by $y=n,n+1,n+2,n+3$ each having $d$ possible outcomes $b\in\{0,\ldots,d-1\}$. Recall that
Alice has $m$ measurements, each with $d$ possible outcomes
$a\in\{0,\ldots,d-1\}$. For this part of the proof, we employ standard self-testing results
based on the Bell--CHSH
inequality~\cite{Yao_self,mckague2012robust,bamps2015sum}. 

For convenience, we introduce the following notation for the
Bell--CHSH expression evaluated within a two-dimensional subspace.
Let $\{s_0,s_1\}$ denote two selected outcomes, and let
$\{x_0,x_1\}$ and $\{y_0,y_1\}$ denote two measurement settings
for Alice and Bob, respectively. 
Associating the outcomes $s_0$ and $s_1$ with the values
$+1$ and $-1$, respectively, we define the corresponding
correlators by
\begin{equation}
    E_c(x,y)
    :=
    \sum_{a,b\in\{s_0,s_1\}}
    \alpha_a\alpha_b\, 
    p(a,b|x,y,c),
    \label{eq:restricted-correlator}
\end{equation}
where $\alpha_{s_0}=1, \alpha_{s_1}=-1.$
The Bell--CHSH expression restricted to these outcomes and
measurement settings is then defined as
\begin{align}
&\mathcal{B}_c
\bigl(
\{s_0,s_1\}|
\{x_0,x_1\},\{y_0,y_1\}
\bigr)
:={} \nonumber \\
& E_c(x_0,y_0)+E_c(x_0,y_1)
+E_c(x_1,y_0)-E_c(x_1,y_1).
\label{eq:Bc-definition}
\end{align}

We first consider an outcome $c$ for which $d_c=2$. Suppose that
there exist two measurements on Alice's side, denoted by $x_0$ and
$x_1$, from the initial set of Alice's measurements $x_0,x_1\in\{0,\ldots,m-1\}$, such that the conditional correlations $\{p(a,b|x,y,c)\}$ satisfy
\begin{eqnarray}
 & \mathcal{B}_{c} \left(\{g(c),g(c)+1\}|\{x_0,x_1\},\{n,n+1\}\right) = 2\sqrt{2} , \nonumber \\ \label{B2} 
\end{eqnarray}
together with $p(c) = \frac2d$. Moreover, conditioned on Carol obtaining outcome $c$, the
probabilities associated with outcomes outside the corresponding subspace vanish, i.e., 
\begin{equation}
    \sum_{a,b\in \{g(c),g(c)+1\}} p(a,b|x,y,c)=1,
\end{equation}
for $x\in\{x_0,x_1\}, y \in\{n,n+1\}$. The above relation implicit $\ket{\psi_c}$ has full support on the outcomes in $\{g(c),g(c)+1\}$ for these measurement settings. Therefore, the standard Bell--CHSH self-testing result then implies Eq.~\eqref{psib-st} with $d_c=2$. 

We next consider an outcome $c$ for which $d_c=3$. Suppose that there
exist three measurements on Alice's side, denoted by $x_0,x_1,x_2$, such that the corresponding conditional correlations
satisfy
\begin{eqnarray}
 & \mathcal{B}_{c} \left(\{g(c),g(c)+1\}|\{x_0,x_1\},\{n,n+1\}\right) = \frac{2}{3} 2\sqrt{2} , \label{B3-1} \\
 &  \mathcal{B}_{c} \left(\{g(c)+1,g(c)+2\}|\{x_1,x_2\},\{n+2,n+3\}\right) = \frac{2}{3}2\sqrt{2},\nonumber \\ \label{B3-2}
 \end{eqnarray}
 and 
 \begin{equation}
      p(g(c)+2,g(c)+2|x,y,c) = \frac13,
 \end{equation}
 for $x \in \{x_0,x_1\}, y \in \{n,n+1\},$
\begin{equation}
    p(g(c),g(c)|x,y,c) = \frac13, \label{B3-3}
\end{equation} 
for $x \in \{x_0,x_2\}, y \in \{n+2,n+3\}.$ 
Moreover, let $p(c)=\frac3d$. 
In addition, Alice and Bob observe
\begin{equation}\label{B3-n}
    \sum_{a,b\in \{g(c),g(c)+1,g(c)+2\}} p(a,b|x,y,c)=1,
\end{equation}
for $x\in\{x_0,x_1,x_2\}, y \in\{n,n+1,n+2,n+3\}$, which means $\ket{\psi_c}$ has full support on the outcomes in $\{g(c),g(c)+1,g(c)+2\}$ for these measurement settings. Together with Eqs.~\eqref{B3-1}--\eqref{B3-3}, this reproduces the correlations required by the self-testing construction of Coladangelo \textit{et al.}~\cite{Coladangelo_NatCom} for the maximally entangled state of local dimension three. Consequently, Eqs.~\eqref{B3-1}--\eqref{B3-n} imply
Eq.~\eqref{psib-st} with $d_c=3$.

We now show that the measurements required above can indeed be chosen from the measurements appearing in the context-projector game on Alice's side, together with the four additional measurements on
Bob's side, which yield the correlations \eqref{B2} and \eqref{B3-1}--\eqref{B3-3}.
First, recall that the merged Peres-24 set contains the computational
basis. Therefore, for every two-dimensional or three-dimensional
projector $\Pi_c$, there exists an Alice measurement containing,
respectively, $\{\ket{g(c)},\ket{g(c)+1}\}$
or $\{\ket{g(c)},\ket{g(c)+1},\ket{g(c)+2}\}.$ We choose this computational-basis measurement as $x_0$.

Furthermore, as can be seen from Table~\ref{tab:peres24}, for every two-dimensional subspace
$\operatorname{span}\{\ket{i},\ket{i+1}\},$
the Peres-24 set contains the two vectors $(\ket{i}\pm \ket{i+1})/\sqrt{2}$. By construction, the same property holds for the merged Peres-24 set
for every relevant value of $i$. Consequently, there exist
measurements setting $x_1$ and $x_2$ such that $x_1$ contains
\[
\left\{
\frac{\ket{g(c)}+\ket{g(c)+1}}{\sqrt{2}},
\frac{\ket{g(c)}-\ket{g(c)+1}}{\sqrt{2}},
\ket{g(c)+2}
\right\},
\]
whereas $x_2$ contains
\[
\left\{
\ket{g(c)},
\frac{\ket{g(c)+1}+\ket{g(c)+2}}{\sqrt{2}},
\frac{\ket{g(c)+1}-\ket{g(c)+2}}{\sqrt{2}}
\right\}.
\] 
We now specify the additional measurements on Bob's side. Define
the Pauli-like operators acting on the corresponding
two-dimensional subspaces by
\begin{eqnarray}
& [\sigma_z]_{g_c} = \ket{g(c)}\!\bra{g(c)} - \ket{g(c)+1}\!\bra{g(c)+1}, \nonumber \\ & [\sigma_z]'_{g_c} = \ket{g(c)+1}\!\bra{g(c)+1} - \ket{g(c)+2}\!\bra{g(c)+2}, \nonumber \\
& [\sigma_x]_{g_c} = \ket{g(c)}\!\bra{g(c)+1} + \ket{g(c)+1}\!\bra{g(c)}, \nonumber \\
& [\sigma_x]'_{g_c} = \ket{g(c)+1}\!\bra{g(c)+2} + \ket{g(c)+2}\!\bra{g(c)+1}.
\end{eqnarray}
Bob's measurements $y=n,n+1$ are defined by the eigenbases of
$([\sigma_z]_{g_c}\pm[\sigma_x]_{g_c})/\sqrt{2}$ on the relevant
two-dimensional subspaces, together with the computational basis on
the orthogonal complement. Explicitly, these measurements are
given by the eigenbases of
\begin{widetext}
\[ \bigoplus_{c \in \kappa_2} \frac{[\sigma_z]_{g_c} \pm [\sigma_x]_{g_c}}{\sqrt{2}} \bigoplus_{c \in \kappa_3} \left(\frac{[\sigma_z]_{g_c} \pm [\sigma_x]_{g_c}}{\sqrt{2}} \oplus \ket{g(c)+2}\!\bra{g(c)+2}\right) \bigoplus_{c\in \overline{\kappa}} \Pi_c.
\] 
\end{widetext}
Here, $\kappa_2,\kappa_3$ are subsets of the set of outcomes $\{0,\cdots,r-1\}$ such that $\kappa_2=\{c:d_c=2\}$,
$\kappa_3=\{c:d_c=3\},$ and
$\overline{\kappa}=\{c:d_c\geq4\}$.
Similarly, Bob's measurements $y=n+2,n+3$ are defined by the
eigenbases of 
\[  \bigoplus_{c \in \kappa_3} \left( \ket{g(c)}\!\bra{g(c)} \oplus \frac{[\sigma_z]'_{g_c} \pm [\sigma_x]'_{g_c}}{\sqrt{2}}\right) \bigoplus_{c\in \overline{\kappa}\cup \kappa_2} \Pi_c,
\] 
One can directly verify that, when applied to the maximally entangled
state, these measurements generate the required Bell--CHSH
correlations: Eq.~\eqref{B2} for $d_c=2$, and
Eqs.~\eqref{B3-1}--\eqref{B3-3} for $d_c=3$. Moreover, these are
precisely the types of measurements appearing in standard
Bell--CHSH self-testing and, for the three-dimensional case, in the
self-testing construction of Coladangelo \textit{et al.}~
\cite{Coladangelo_NatCom} for $\ket{\phi^+_3}$. Hence,
Eq.~\eqref{psib-st} follows also for $d_c=2$ and $d_c=3$.

Finally, note that, since $\Pi_c$ is a rank-$d_c$ projector,
\begin{equation}
    (\id\otimes\Pi_c)
    \ket{\phi_d^{+}}\!\bra{\phi_d^{+}}
    (\id\otimes\Pi_c)
    =
    \frac{d_c}{d}\,
    \ket{\phi_{d_c}^{+}}\!\bra{\phi_{d_c}^{+}} .
\label{eq:projected-max-ent-state}
\end{equation}
Substituting $\widehat{\rho}_{AB_c}$, given by \eqref{post-state-c}, 
into Eq.~\eqref{psib-st} and using
Eq.~\eqref{eq:projected-max-ent-state}, we obtain precisely the post-measurement condition
Eq.~\eqref{stABcall} in
Definition~\ref{def:instrument-self-testing}. Since the argument holds separately for every outcome $c$, with an outcome-dependent extraction
isometry $U_c$, the claimed instrument self-testing statement follows. This completes the proof.
\end{proof}

%%%%%%%%%%%%%%%%%%%%%%%%%%%%%%%%%%%%%%%%%%%%%%%%%%%%%%%%%%%%%%%%%%

\subsection{Self-testing ideal measurement based on Peres-$84$ set of vectors in $8$-dimensional space}

%%%%%%%%%%%%%%%%%%%%%%%%%%%%%%%%%%%%%%%%%%%%%%%%%%%%%%%%%%%%%%%%%% 

Peres-84 set of vectors is obtained by merging five Peres-24 vectors in 4-dimensional subspaces. The vectors are given below.
\begin{widetext}
\begin{center}
\begin{table}[h!]
{\scriptsize
\begin{tabular}{cccccccccccccccccccccccccccccccccccccccccc}
\hline\hline
0 & 1 & 2 & 3 & 4 & 5 & 6 & 7 & 8 & 9 & 10 & 11 & 12 & 13 & 14 & 15 & 16 & 17 & 18 & 19 & 20 & 21 & 22 & 23 & 24 & 25 & 26 & 27 & 28 & 29 & 30 & 31 & 32 & 33 & 34 & 35 & 36 & 37 & 38 & 39 & 40 & 41 \\
\hline
 0 & 0 & 0 & 0 & 0 & 0 & 0 & 0 & 0 & 0 & 0 & 0 & 0 & 0 & 0 & 0 & 0 & 0 & 0 & 0 & 0 & 0 & 0 & 0 & 0 & 0 & 0 & 0 & 0 & 0 & 0 & 0 & 0 & 0 & 0 & 0 & 0 & 0 & 0 & 0 & 0 & 0 \\
 0 & 0 & 0 & 0 & 0 & 0 & 0 & 0 & 0 & 0 & 0 & 0 & 0 & 0 & 0 & 0 & 0 & 0 & 0 & 0 & 0 & 0 & 0 & 0 & 0 & 0 & 0 & 0 & 0 & 0 & 0 & 0 & 0 & 0 & 0 & 0 & 0 & 0 & 0 & 0 & 0 & 0 \\
 0 & 0 & 0 & 0 & 0 & 0 & 0 & 0 & 0 & 0 & 0 & 0 & 0 & 0 & 0 & 0 & 0 & 0 & 0 & 0 & 0 & 0 & 0 & 0 & 0 & 0 & 0 & 0 & 0 & 0 & 0 & 0 & 0 & 0 & 0 & 0 & 0 & 0 & 0 & 1 & 1 & 1 \\
 0 & 0 & 0 & 0 & 0 & 0 & 0 & 0 & 0 & 0 & 0 & 0 & 0 & 0 & 0 & 0 & 0 & 0 & 0 & 0 & 0 & 0 & 0 & 0 & 1 & 1 & 1 & 1 & 1 & 1 & 1 & 1 & 1 & 1 & 1 & 1 & 1 & 1 & 1 & 0 & 1 & 1 \\
 1 & 0 & 0 & 0 & 1 & 1 & 1 & 1 & 1 & 1 & 1 & 1 & 1 & 1 & 0 & 0 & 1 & 1 & 0 & 0 & 1 & 1 & 0 & 0 & 0 & 1 & 1 & $\bar{1}$ & $\bar{1}$ & 1 & 1 & $\bar{1}$ & $\bar{1}$ & 1 & $\bar{1}$ & 0 & 0 & 0 & 0 & 0 & 1 & $\bar{1}$ \\
 0 & 1 & 0 & 0 & 1 & 1 & $\bar{1}$ & $\bar{1}$ & 1 & 1 & $\bar{1}$ & $\bar{1}$ & 1 & $\bar{1}$ & 0 & 0 & 0 & 0 & 1 & 1 & 0 & 0 & 1 & 1 & 0 & 1 & $\bar{1}$ & 1 & $\bar{1}$ & 1 & $\bar{1}$ & 1 & $\bar{1}$ & 0 & 0 & 1 & $\bar{1}$ & 0 & 0 & 0 & 1 & $\bar{1}$ \\
 0 & 0 & 1 & 0 & 1 & $\bar{1}$ & 1 & $\bar{1}$ & 1 & $\bar{1}$ & 1 & $\bar{1}$ & 0 & 0 & 1 & 1 & 1 & $\bar{1}$ & 0 & 0 & 0 & 0 & 1 & $\bar{1}$ & 0 & 1 & $\bar{1}$ & $\bar{1}$ & 1 & $\bar{1}$ & 1 & 1 & $\bar{1}$ & 0 & 0 & 0 & 0 & 1 & $\bar{1}$ & 0 & 0 & 0 \\
 0 & 0 & 0 & 1 & 1 & $\bar{1}$ & $\bar{1}$ & 1 & $\bar{1}$ & 1 & 1 & $\bar{1}$ & 0 & 0 & 1 & $\bar{1}$ & 0 & 0 & 1 & $\bar{1}$ & 1 & $\bar{1}$ & 0 & 0 & 0 & 0 & 0 & 0 & 0 & 0 & 0 & 0 & 0 & 0 & 0 & 0 & 0 & 0 & 0 & 0 & 0 & 0 \\
 \hline\hline
\end{tabular}
\vspace{1em}
\begin{tabular}{cccccccccccccccccccccccccccccccccccccccccc}
\hline\hline
42 & 43 & 44 & 45 & 46 & 47 & 48 & 49 & 50 & 51 & 52 & 53 & 54 & 55 & 56 & 57 & 58 & 59 & 60 & 61 & 62 & 63 & 64 & 65 & 66 & 67 & 68 & 69 & 70 & 71 & 72 & 73 & 74 & 75 & 76 & 77 & 78 & 79 & 80 & 81 & 82 & 83 \\ \hline
 0 & 0 & 0 & 0 & 0 & 0 & 0 & 0 & 0 & 0 & 0 & 0 & 0 & 0 & 0 & 0 & 0 & 0 & 0 & 0 & 0 & 0 & 0 & 0 & 0 & 0 & 0 & 1 & 1 & 1 & 1 & 1 & 1 & 1 & 1 & 1 & 1 & 1 & 1 & 1 & 1 & 1 \\ 
 0 & 0 & 0 & 0 & 0 & 0 & 0 & 0 & 0 & 0 & 0 & 0 & 1 & 1 & 1 & 1 & 1 & 1 & 1 & 1 & 1 & 1 & 1 & 1 & 1 & 1 & 1 & 0 & 1 & 1 & $\bar{1}$ & $\bar{1}$ & 1 & 1 & $\bar{1}$ & $\bar{1}$ & 1 & $\bar{1}$ & 0 & 0 & 0 & 0 \\
 1 & 1 & 1 & 1 & 1 & 1 & 1 & 1 & 1 & 1 & 1 & 1 & 0 & 1 & 1 & $\bar{1}$ & $\bar{1}$ & 1 & 1 & $\bar{1}$ & $\bar{1}$ & 1 & $\bar{1}$ & 0 & 0 & 0 & 0 & 0 & 1 & $\bar{1}$ & 1 & $\bar{1}$ & 1 & $\bar{1}$ & 1 & $\bar{1}$ & 0 & 0 & 1 & $\bar{1}$ & 0 & 0 \\
 $\bar{1}$ & $\bar{1}$ & 1 & 1 & $\bar{1}$ & $\bar{1}$ & 1 & $\bar{1}$ & 0 & 0 & 0 & 0 & 0 & 1 & $\bar{1}$ & 1 & $\bar{1}$ & 1 & $\bar{1}$ & 1 & $\bar{1}$ & 0 & 0 & 1 & $\bar{1}$ & 0 & 0 & 0 & 1 & $\bar{1}$ & $\bar{1}$ & 1 & $\bar{1}$ & 1 & 1 & $\bar{1}$ & 0 & 0 & 0 & 0 & 1 & $\bar{1}$ \\
 1 & $\bar{1}$ & 1 & $\bar{1}$ & 1 & $\bar{1}$ & 0 & 0 & 1 & $\bar{1}$ & 0 & 0 & 0 & 1 & $\bar{1}$ & $\bar{1}$ & 1 & $\bar{1}$ & 1 & 1 & $\bar{1}$ & 0 & 0 & 0 & 0 & 1 & $\bar{1}$ & 0 & 0 & 0 & 0 & 0 & 0 & 0 & 0 & 0 & 0 & 0 & 0 & 0 & 0 & 0 \\
 $\bar{1}$ & 1 & $\bar{1}$ & 1 & 1 & $\bar{1}$ & 0 & 0 & 0 & 0 & 1 & $\bar{1}$ & 0 & 0 & 0 & 0 & 0 & 0 & 0 & 0 & 0 & 0 & 0 & 0 & 0 & 0 & 0 & 0 & 0 & 0 & 0 & 0 & 0 & 0 & 0 & 0 & 0 & 0 & 0 & 0 & 0 & 0 \\
 0 & 0 & 0 & 0 & 0 & 0 & 0 & 0 & 0 & 0 & 0 & 0 & 0 & 0 & 0 & 0 & 0 & 0 & 0 & 0 & 0 & 0 & 0 & 0 & 0 & 0 & 0 & 0 & 0 & 0 & 0 & 0 & 0 & 0 & 0 & 0 & 0 & 0 & 0 & 0 & 0 & 0 \\
 0 & 0 & 0 & 0 & 0 & 0 & 0 & 0 & 0 & 0 & 0 & 0 & 0 & 0 & 0 & 0 & 0 & 0 & 0 & 0 & 0 & 0 & 0 & 0 & 0 & 0 & 0 & 0 & 0 & 0 & 0 & 0 & 0 & 0 & 0 & 0 & 0 & 0 & 0 & 0 & 0 & 0\\
 \hline\hline
\end{tabular}
}
\caption{\label{tab:peres84} Peres-84. Each column represents the unnormalized vector $|v_i\rangle$, with $\bar{1}=-1$.
}
\end{table}
\end{center}
\end{widetext}
From the correlations $\{p(a,b|x,y)\}$, the maximally entangled state $\ket{\phi_8^+}$ is self-tested by observing the perfect strategy $\mathcal{B}^\text{P}_8=1$. To self-test the ideal measurement $\{\sum_{i=0}^3\ket{i}\!\bra{i}, \sum_{i=4}^7\ket{i}\!\bra{i}\}$, the parties verify perfect strategies for two context-projector games defined within appropriate subspaces.
For outcome $c=1$, they consider the context-projector game constructed from the Peres-24 vectors $\{\ket{v_i}\}_{i=0}^{23}$. Achieving the perfect strategy self-tests the projector $\sum_{i=4}^7\ket{i}\!\bra{i}$. 
Similarly, when outcome $c=0$ is obtained, they consider the context-projector game based on the KS set of 24 vectors $\left\{\ket{v_{24}},\ket{v_{39}},\ket{v_{48}},\ket{v_{49}},\ket{v_{54}}, \{\ket{v_i}\}_{i=63}^{66}, \{\ket{v_i}\}_{i=69}^{83}\right\}$. The corresponding perfect strategy then self-tests the projector $\sum_{i=0}^3\ket{i}\!\bra{i}$.

%%%%%%%%%%%%%%%%%%%%%%%%%%%%%%%%%%%%%%%%%%%%%%%%%%%%%%%%%%%%%%%%%%%

\subsection*{Robustness of the self-testing}

%%%%%%%%%%%%%%%%%%%%%%%%%%%%%%%%%%%%%%%%%%%%%%%%%%%%%%%%%%%%%%%%%%

We now discuss about the robustness of the self-testing for an instrument element. Since a quantum operation is completely characterized by
its Choi operator, it is natural to compare the Choi representation
of the extracted physical operation with that of the corresponding
ideal L\"uders operation.

For positive semidefinite operators $\rho$ and $\sigma$ with
$\tr(\rho),\tr(\sigma)\leq 1$, we use the generalized trace distance $D_{\rm g} (\rho,\sigma):= \frac{1}{2}\|\rho-\sigma\|_1 +
 \frac{1}{2} \left|\tr(\rho)-\tr(\sigma)\right|.$
For normalized states, this reduces to the usual trace distance. For a fixed outcome $c$, we define the robustness of the physical instrument element $\mathcal{M}_c$ with respect to the ideal L\"uders operation
$\mathcal{P}_c(X)=\Pi_cX\Pi_c$ as
\begin{widetext}
    \begin{equation} \mathcal{R}_c(\mathcal{M}_c,\mathcal{P}_c)
    := \inf_{U_c,\rho_{\mathrm{aux},c}}
    D_{\rm g}\Bigg( (\id_A\otimes\mathcal{U}_c)
    \circ(\id_A\otimes\mathcal{M}_c)
    \circ(\id_A\otimes\mathcal{V})
    \left[
        \ket{\phi_d^+}\!\bra{\phi_d^+}
        \otimes\rho_{\mathrm{aux}} \right],
        (\id\otimes\Pi_c) \ket{\phi_d^+}\!\bra{\phi_d^+}
    (\id\otimes\Pi_c)
    \otimes\rho_{\mathrm{aux},c}  \Bigg),
    \label{instrument-robustness}
\end{equation}
\end{widetext}
where $\mathcal{U}_c(X)=U_cXU_c^\dagger$, $\mathcal{V}(X)=VXV^\dagger.$ The first argument of Eq.~\eqref{instrument-robustness} is the Choi representation of the physical outcome operation in the self-tested reference frame, whereas the second is the Choi operator of the corresponding ideal L\"uders operation. 
Since these operators are subnormalized, this distance captures both the probability associated with outcome $c$ and the corresponding post-measurement update.

The following observation relates this quantity to the robustness
of the two self-testing relations in
Definition~\ref{def:instrument-self-testing}.

\begin{observation}
\label{lem:rob}
Suppose that there exist local extraction isometries $U_A,U_B$, and $U_c$ such that
\begin{equation}
    D\left(
        (\mathcal{U}_A\otimes\mathcal{U}_B)
        \bigl(\rho\bigr),
        \ket{\phi_d^+}\!\bra{\phi_d^+}
        \otimes\rho_{\mathrm{aux}}
    \right)
    \leq\delta_1,
    \label{del1}
\end{equation}
and
\begin{equation}
\begin{split}
    D_{\rm g}\Bigg(
        &(\mathcal{U}_A\otimes\mathcal{U}_c)
        \left[
            (\id_A\otimes\mathcal{M}_c)
            \bigl(\rho\bigr)
        \right],
        \\
        &(\id\otimes\Pi_c)
        \ket{\phi_d^+}\!\bra{\phi_d^+}
        (\id\otimes\Pi_c)
        \otimes\rho_{\mathrm{aux},c}
    \Bigg)
    \leq\delta_2.
    \label{del2}
\end{split}
\end{equation}
Then, taking $V=U_B^\dagger$, the robustness of the instrument
element satisfies
\begin{equation}
    \mathcal{R}_c(\mathcal{M}_c,\mathcal{P}_c)
    \leq
    \delta_1+\delta_2.
    \label{robustness-bound}
\end{equation}
\end{observation}
\begin{proof}[Proof of Observation~\ref{lem:rob}]
For compactness, define
\begin{equation}
    \rho_{\mathrm{id}}
    :=
    \ket{\phi_d^+}\!\bra{\phi_d^+}
    \otimes\rho_{\mathrm{aux}},
\end{equation}
and
\begin{equation}
    \rho_{\mathrm{ext}}
    :=
    (\mathcal{U}_A\otimes\mathcal{U}_B)
    \bigl(\rho\bigr).
\end{equation}
Equation~\eqref{del1} then reads
\begin{equation}
    D(\rho_{\mathrm{ext}},\rho_{\mathrm{id}})
    \leq\delta_1.
    \label{del1-short}
\end{equation}

Since $V=U_B^\dagger$, applying
$\id_A\otimes\mathcal V$ to the Bob part of the extracted
representation maps it back to the physical input representation of
Carol's instrument. We now apply the completely positive
trace-nonincreasing map corresponding to Carol's outcome $c$,
followed by the output extraction $\mathcal U_c$. By contractivity
of the generalized trace distance under completely positive
trace-nonincreasing maps, Eq.~\eqref{del1-short} implies
\begin{equation}
\begin{split}
    D_{\rm g}\Bigg(
        &(\id_A\otimes\mathcal U_c)
        (\id_A\otimes\mathcal M_c)
        (\id_A\otimes\mathcal V)
        [\rho_{\mathrm{id}}],
        \\
        &(\mathcal U_A\otimes\mathcal U_c)
        \left[
            (\id_A\otimes\mathcal M_c)
            \bigl(\rho\bigr)
        \right]
    \Bigg)
    \leq\delta_1.
    \label{rob-step1}
\end{split}
\end{equation}
Here we have used the fact that the physical operation on Bob's
system commutes with Alice's local extraction $\mathcal U_A$.

Next, define the ideal unnormalized post-measurement Choi operator
\begin{equation}
    \tau_c^{\mathrm{id}}
    :=
    (\id\otimes\Pi_c)
    \ket{\phi_d^+}\!\bra{\phi_d^+}
    (\id\otimes\Pi_c)
    \otimes\rho_{\mathrm{aux},c}.
\end{equation}
By assumption, Eq.~\eqref{del2} gives
\begin{equation}
\begin{split}
    D_{\rm g}\Bigg(
        &(\mathcal U_A\otimes\mathcal U_c)
        \left[
            (\id_A\otimes\mathcal M_c)
            \bigl(\rho\bigr)
        \right],
        \tau_c^{\mathrm{id}}
    \Bigg)
    \leq\delta_2.
    \label{rob-step2}
\end{split}
\end{equation}

Applying the triangle inequality to Eqs.~\eqref{rob-step1} and \eqref{rob-step2}, we obtain
\begin{equation}
    D_{\rm g}\Bigg(
        (\id_A\otimes\mathcal U_c)
        (\id_A\otimes\mathcal M_c)
        (\id_A\otimes\mathcal V)
        [\rho_{\mathrm{id}}],
        \tau_c^{\mathrm{id}}
    \Bigg) \leq
    \delta_1+\delta_2.
\end{equation}
Since $\mathcal{R}_c(\mathcal{M}_c,\mathcal{P}_c)$ is defined as
the infimum of this distance over the admissible output extraction
$U_c$ and auxiliary state $\rho_{\mathrm{aux},c}$, the particular
$U_c$ and $\rho_{\mathrm{aux},c}$ satisfying Eq.~\eqref{del2}
provide an admissible choice. Therefore, Eq.~\eqref{robustness-bound} holds.
\end{proof}
 
\end{document}